\documentclass[11pt]{article}
\usepackage{amsmath,amsthm,amsfonts,amssymb}
\usepackage{fullpage}
\usepackage{todonotes}
\usepackage[colorlinks=true]{hyperref}
\newtheorem{theorem}{Theorem}
\newtheorem{lemma}{Lemma}

\newtheorem{corollary}{Corollary}

\theoremstyle{definition}
\newtheorem{definition}{Definition}
\theoremstyle{plain}

\newcommand{\bbN}{\mathbb{N}}
\newcommand{\bit}{\{0,1\}}
\newcommand{\bv}{\mathbf{v}}

\newcommand{\bx}{\mathbf{x}}
\newcommand{\by}{\mathbf{y}}
\newcommand{\bz}{\mathbf{z}}

\newcommand{\calP}{\mathcal{P}}

\newcommand{\calR}{\mathcal{R}}
\newcommand{\calS}{\mathcal{S}}
\newcommand{\calT}{\mathcal{T}}
\newcommand{\calX}{\mathcal{X}}
\newcommand{\calY}{\mathcal{Y}}

\newcommand{\dist}{\mathrm{d}}
\newcommand{\E}{\mathop{\mathrm{E}}}
\newcommand{\F}{\mathbb{F}}

\newcommand{\no}{\mathrm{no}}
\newcommand{\NonIden}{\textsc{NotEq}}
\newcommand{\set}[1]{\{#1\}}

\newcommand{\yes}{\mathrm{yes}}

\title{A Characterization of Constant-Sample Testable Properties}

\author{Eric Blais\thanks{Supported by an NSERC Discovery Grant}\\
David R. Cheriton School of Computer Science \\
University of Waterloo \\
\texttt{eric.blais@uwaterloo.ca}
\and
Yuichi Yoshida\thanks{Supported by JSPS Grant-in-Aid for Young Scientists (B) (No.~26730009), MEXT Grant-in-Aid for Scientific Research on Innovative Areas (24106003), and JST, ERATO, Kawarabayashi Large Graph Project.}\\
  National Institute of Informatics
  \emph{and}\\
  Preferred Infrastructure, Inc.\\
  \texttt{yyoshida@nii.ac.jp}
}

\begin{document}
\maketitle
\begin{abstract}
We characterize the set of properties of Boolean-valued functions on a finite domain $\mathcal{X}$ that are testable with a constant number of samples.
  Specifically, we show that a property $\mathcal{P}$ is testable with a constant number of samples if and only if it is (essentially) a $k$-part symmetric property for some constant $k$,
  where a property is \emph{$k$-part symmetric} if there is a partition $S_1,\ldots,S_k$ of $\mathcal{X}$ such that whether $f:\mathcal{X} \to \{0,1\}$ satisfies the property is determined solely by the densities of $f$ on $S_1,\ldots,S_k$.

  We use this characterization to obtain a number of corollaries, namely:
  \begin{itemize}
  \item A graph property $\mathcal{P}$ is testable with a constant number of samples if and only if whether a graph $G$ satisfies $\mathcal{P}$ is (essentially) determined by the edge density of $G$.
  \item An affine-invariant property $\mathcal{P}$ of functions $f:\mathbb{F}_p^n \to \{0,1\}$ is testable with a constant number of samples if and only if whether $f$ satisfies $\mathcal{P}$ is (essentially) determined by the density of $f$.
  \item For every constant $d \geq 1$, monotonicity of functions $f : [n]^d \to \{0, 1\}$ on the $d$-dimensional hypergrid is testable with a constant number of samples.
  \end{itemize}
\end{abstract}

%!TEX root=SampleTestableCharacterization.tex

\section{Introduction}

Property testing~~\cite{RubinfeldS96,Goldreich:1998wa} is concerned with the general question: for which properties of combinatorial objects can we efficiently distinguish the objects that have the property from those that are ``far'' from having the same property? Consideration of this question has led to surprisingly powerful results: many natural graph properties~\cite{Goldreich:1998wa,Alon:2009gn}, algebraic properties of functions on finite fields~\cite{RubinfeldS96,Yoshida:2014tq}, and structural properties of Boolean functions~\cite{FKRSS04,DLM+07}, for example, can be tested with a \emph{constant} number of queries to the object being tested.
Nearly all of these results appear to rely critically on the testing algorithm's ability to query the unknown object at locations of its choosing. The goal of the present work is to determine to what extent this is actually the case: is it possible that some (or most!) of these properties can still be tested efficiently even if the tester has \emph{no} control over the choice of queries that it makes?

Our main question is made precise in the \emph{sample-based} property testing model originally introduced by Goldreich, Goldwasser, and Ron~\cite{Goldreich:1998wa}.
For a finite set $\calX$, let $\{0,1\}^\calX$ denote the set of Boolean-valued
functions on $\calX$ endowed with the normalized Hamming distance metric
$\dist(f,g) := |\{ x \in \calX : f(x) \neq g(x)\}|/|\calX|$.
A \emph{property} of functions mapping $\calX$ to $\{0,1\}$ is a subset
$\calP \subseteq \{0,1\}^\calX$.
A function is \emph{$\epsilon$-close} to $\calP$ if it is in the set
$\calP_\epsilon := \{f \in \{0,1\}^\calX : \exists g \in \calP \mbox{ s.t. }
\dist(f,g) \le \epsilon\}$; otherwise, it is \emph{$\epsilon$-far} from $\calP$.
% The fundamental
% \emph{property testing}~\cite{RubinfeldS96,Goldreich:1998wa} question asks:
% for which properties can we efficiently
% distinguish functions in $\calP$ from those that are $\epsilon$-far from $\calP$?
%We consider the property testing question in the \emph{sample-based} model.
An \emph{$s$-sample $\epsilon$-tester} for a property $\calP$
is a randomized algorithm with bounded error 
that observes $s$ pairs $(x_1,f(x_1)),\ldots,(x_s,f(x_s)) \in \calX \times \{0,1\}$ with
$x_1,\ldots,x_s$ drawn independently and uniformly at random from $\calX$
and then accepts when $f \in \calP$ and rejects when $f$ is
$\epsilon$-far from $\calP$.
The \emph{sample complexity} of $\calP$ for some $\epsilon > 0$ is the minimum value of
$s$ such that it has an $s$-sample $\epsilon$-tester.
When the sample complexity of $\calP$
is independent of the domain size $|\calX|$ of the functions for every $\epsilon > 0$, we say that
$\calP$ is \emph{constant-sample testable}.

In the 20 years since the original introduction of the sample-based property testing model,
a number of different properties have been shown to be constant-sample testable,
including all symmetric properties (folklore; see the discussion in~\cite{KaufmanS08}),
unions of intervals~\cite{Kearns:2000ev}, decision trees over
low-dimensional domains~\cite{Kearns:2000ev},
and convexity of images~\cite{ Raskhodnikova03,Berman:2015tb}.
Conversely, many other properties, such as
monotonicity of Boolean functions~\cite{Goldreich:2000ke},
linearity~\cite{ AlonHW13,Goldreich:2015fh},
linear threshold functions~\cite{Balcan:2012ew}, and
$k$-colorability of graphs~\cite{Goldreich:2015fh}, are known to not be
constant-sample testable.
Yet, despite a wide-spread belief that most properties are not testable with a constant number of samples, it remained open until now to determine whether this is actually the case or not. Our main result settles this question, and in the process unifies all of the above results and explains what makes various properties testable with a constant number of samples or not.
% The purpose of the present research is to unify these results by
% identifying which properties are constant-sample testable and by
% explaining what makes different properties constant-sample testable or not.

\subsection{Main result}

We show that constant-sample testability is closely tied to a
particular notion of symmetry---or invariance---of properties.
Let $\calS_\calX$ denote the set of permutations on a finite set $\calX$, and for any subset $S \subseteq \calX$, let $\calS_{\calX}^{(S)}$ denote the set of permutations on $\calX$ that preserves the elements in $S$.
A permutation $\pi \in \calS_\calX$ acts on functions $f : \calX \to \{0,1\}$
in the obvious way: $\pi f$ is the function that satisfies
$(\pi f)(x) = f( \pi x)$ for every $x \in \calX$.
The property $\calP \subseteq \{0,1\}^\calX$ is
\emph{invariant} under a permutation $\pi \in \calS_\calX$
if for every $f \in \calP$, we also have $\pi f \in \calP$.
$\calP$ is \emph{(fully) symmetric} if it is invariant
under all permutations in $\calS_\calX$.
The following definition relaxes this condition to obtain a notion of
``partial'' symmetry.

\begin{definition}
  For any $k \in \bbN$, the property $\calP \subseteq \bit^\calX$ is
  \emph{$k$-part symmetric} if there is a partition  of $\calX$ into
  $k$ parts $X_1, \ldots, X_k$ such that $\calP$ is invariant under
  all the permutations in $\calS_{X_1} \times \cdots \times \calS_{X_k}$,
  where $\calS_{X_i}$ is the set of permutations over $\calX$ that
  preserves elements in $X_i$.
\end{definition}

Equivalently, $\calP$ is $k$-part symmetric if there exists a partition
$X_1,\ldots,X_k$ of $\calX$ such that the event $f \in \calP$ is completely
determined by the density $\frac{|f^{-1}(1) \cap X_i|}{|X_i|}$ of $f$ in each of the
sets $X_1,\ldots,X_k$.
Our main result shows that $O(1)$-part symmetry is also essentially equivalent
to constant-sample testability.

\begin{theorem}\label{thm:characterization}
  The property $\calP$ of a function $f:\calX \to \bit$ is constant-sample testable if and only if for any $\epsilon > 0$, there exists a constant $k = k_\calP(\epsilon)$ that is independent of $\calX$ and a $k$-part symmetric property $\calP'$
  such that $\calP \subseteq \calP' \subseteq \calP_{\epsilon}$.
\end{theorem}

In words, Theorem~\ref{thm:characterization} says that constant-sample testable
properties are the properties $\calP$ that can be covered by some $O(1)$-part symmetric
property $\calP'$ that does not include any function that is $\epsilon$-far from
$\calP$. Note that this characterization cannot be replaced with the condition
that $\calP$ itself is $k$-part symmetric. To see this, consider the
\emph{function non-identity} property $\NonIden(g)$ that includes every function
except some non-constant function $g : \calX \to \bit$. This property is
not $k$-part symmetric for any $k = O(1)$, but the trivial algorithm that accepts
every function is a valid $\epsilon$-tester for $\NonIden(g)$ for any constant
$\epsilon > 0$.

Theorem~\ref{thm:characterization} can easily be generalized to apply to properties of functions mapping $\calX$ to any finite set $\calY$.
We restrict our attention to the range $\calY = \bit$ for simplicity and clarity of presentation.
The sample-based property testing model is naturally extended to 
non-uniform distributions over the input domain. 
It appears likely that Theorem~\ref{thm:characterization} can be generalized to
this more general setting as well, though we have not attempted to do so.

The proof of Theorem~\ref{thm:characterization} follows the general outline
of previous characterizations of the properties testable in the query-based model 
(e.g., \cite{Alon:2009gn,Yoshida:2014tq,Bhattacharyya:2013fa}). 
As with those results, 
the most interesting part of the proof is the direction showing that
constant-sample testability implies coverage by an $O(1)$-part symmetric property, 
and this proof is established with a regularity lemma.
Our proof departs from previous results in both the type of regularity lemma
that we use and in how we use it. 
For a more detailed discussion of the proof, see Section~\ref{sec:overview}.

\subsection{Applications}

The characterization of constant-sample testability in Theorem~\ref{thm:characterization} can be used to derive a number of different corollaries.
We describe a few of these.

When $\calX$ is identified with the $\binom{n}{2}$ pairs of vertices in $V$, the
function $f : \calX \to \{0,1\}$ represents a graph $G = (V,E)$ where $E = f^{-1}(1)$.
A \emph{graph property} is a property of these functions that is invariant under
relabelling of the vertices.
Observant readers will have noted that no non-symmetric graph property
was present in the list of properties that have been determined to be
constant-sample testable.
Using Theorem~\ref{thm:characterization}, we can show that this is
unavoidable: the only graph properties that are constant-sample testable are those
that are (essentially) fully symmetric.

\begin{corollary}
\label{cor:graph}
  For every $\epsilon > 0$, if $\calP$ is a graph property that
  is $\epsilon$-testable with a constant number of samples, then there is a symmetric property
  $\calP'$ such that $\calP \subseteq \calP' \subseteq \calP_\epsilon$.
\end{corollary}

When $\calX$ is identified with a finite field $\F_p^n$, a property
$\calP \subseteq \{0,1\}^\calX$ is \emph{affine-invariant} if it is invariant
under any affine transformation over $\F_p^n$.
Theorem~\ref{thm:characterization} can be used to show that
symmetric properties are essentially the only
constant-sample testable affine-invariant properties.

\begin{corollary}
\label{cor:affine}
  For every $\epsilon > 0$, if $\calP$ is an affine-invariant property of functions $f : \F_p^n \to \{0,1\}$ that
  is $\epsilon$-testable with a constant number of samples, then there is a symmetric property
  $\calP'$ such that $\calP \subseteq \calP' \subseteq \calP_\epsilon$.
\end{corollary}

These two corollaries show that Theorem~\ref{thm:characterization} can be
used to show that some properties are not constant-sample testable. In our
third application, we show that the other direction of the characterization
can also be used to show that some properties are constant-sample testable.

Fix a constant $d \ge 1$.
Two points $x,y \in [n]^d := \{1,2,\ldots,n\}^d$ satisfy $x \preceq y$ when
$x_1 \le y_1$, $\ldots$, and $x_d \le y_d$.
The function $f : [n]^d \to \bit$ is \emph{monotone}
if for every $x \preceq y \in [n]^d$, we have $f(x) \le f(y)$.
When $d=1$, it is folklore knowledge that monotonicity of Boolean-valued
functions on the line is constant-sample testable.
Using Theorem~\ref{thm:characterization}, we give an easy proof showing that the same holds for every other constant dimension $d$.

\begin{corollary}
\label{cor:monotonicity}
  For every constant $d \ge 1$ and constant $\epsilon > 0$, we can $\epsilon$-test monotonicity of
  functions $f : [n]^d \to \{0,1\}$ on the $d$-dimensional hypergrid
  with a constant number of samples.
\end{corollary}

Chen, Servedio, and Tan~\cite{Chen:2014wj} showed that the number of
queries (and thus also of samples) required to test monotonicity
of $f : [n]^d \to \{0,1\}$ must depend on $d$. (The same result for the case
where $n=2$ was first established by Fischer et al.~\cite{Fischer:2002ev}.)
Combined with the result above,
this shows that monotonicity of Boolean-valued functions on the hypergrid is
constant-sample testable if and only if the number of dimensions of the
hypergrid is constant.

\subsection{Related work}

\paragraph{Sample-based property testing.}

The first general result regarding constant-sample testable
properties goes back to the original work of Goldreich, Goldwasser, and Ron~\cite{Goldreich:1998wa}.  They showed that every property with
constant VC dimension (and, more generally, every property that corresponds to a class of functions that can be properly learned with a constant number of samples)
is constant-sample testable. As they also show, this condition is not necessary for constant-sample testability---in fact, there are even properties that are testable with a constant number of samples whose corresponding class require
a \emph{linear} number of samples to learn~\cite[Prop.~3.3.1]{Goldreich:1998wa}.

More general results on sample-based testers were obtained by
Balcan~\emph{et~al.}~\cite{Balcan:2012ew}. In particular, they defined a notion of 
\emph{testing dimension} of a property $\calP$ in terms of the total
variation distance between the distributions on the tester's observations 
when a function is drawn from
distributions $\pi_{\yes}$ and $\pi_{\no}$ essentially
supported on $\calP$ and $\overline{\calP_\epsilon}$, respectively.
They show that this testing dimension captures the sample complexity
of $\calP$ up to constant factors, and observe that it can be interpreted
as an ``average VC dimension''-type of complexity measure.
It would be interesting to see whether the combinatorial characterization in
Theorem~\ref{thm:characterization} could be combined with these results to
offer new insights into the connections between invariance and VC dimension-like
complexity measures.

Finally, Goldreich and Ron~\cite{Goldreich:2015fh} and Fischer~\emph{et~al.}~\cite{Fischer:2014dw,Fischer:2015vg} established connections between the query- and sample-based models of property testing giving sufficient conditions for \emph{sublinear}-sample testability
of properties. The exact bounds between sample complexity and partial symmetry in the proof of Theorem~\ref{thm:characterization} yield another sufficient condition for sublinear-sample testability: every property $\calP$ that can be $\epsilon$-covered 
by an $o(\log \log |\calX|)$-part symmetric function $\calP'$ has sublinear sample 
complexity $o(|\calX|)$. As far as we can tell, these two characterizations are incomparable.

\paragraph{Symmetry and testability.}

The present work was heavily influenced by the
systematic exploration of connections between the invariances 
of properties and their testability initiated by 
Kaufman and Sudan~\cite{ KaufmanS08}. (See also~\cite{ Sudan10}.)
In that work, the authors showed that such connections yield new
insights into the testability of algebraic properties in the
query-based property testing model, and advocated for further 
study of the invariance of properties as a means to better understand
their testability.
Theorem~\ref{thm:characterization} provides evidence that this approach
is a critical tool in the study of sample-based property testing model
as well.

The notion of partial symmetry and its connections to computational efficiency has a long history---it goes back at least to the pioneering work of Shannon~\cite{ Shannon49}.
Partial symmetry also appeared previously in a property testing context
in the authors' joint work with Amit Weinstein on characterizing the set of functions for which isomorphism testing is constant-query testable~\cite{ BlaisWY15}.
However, it should be noted that the notion of partial symmetry
considered in~\cite{ BlaisWY15} does \emph{not} correspond to the notion of
$k$-part symmetry studied here. In fact, as mentioned in the conference version of that paper, there are $2$-part symmetric functions for which isomorphism testing is not constant-query testable, so the two characterizations inherently require different notions of partial symmetry.

\subsection{Organization}

The proof of Theorem~\ref{thm:characterization} is presented in
Section~\ref{sec:characterization}.
The proofs of the application results are in Section~\ref{sec:applications}.
Finally, since the weak regularity lemma that we use in the proof of Theorem~\ref{thm:characterization} is not completely standard, we include its proof in Section~\ref{sec:regularity} for completeness.

%!TEX root=SampleTestableCharacterization.tex

\section{Proof of Theorem~\ref{thm:characterization}}
\label{sec:characterization}

We prove the two parts (sufficiency and necessity) of Theorem~\ref{thm:characterization} in Sections~\ref{sec:sufficiency} and~\ref{sec:necessity}, respectively.
But first, we provide a high-level overview of the proof and discuss
the connections with regularity lemmas in Section~\ref{sec:overview}.

\subsection{Overview of the proof}
\label{sec:overview}

\paragraph{Symmetry implies testability.}
 The proof of this direction of the theorem is straightforward and is obtained by generalizing the following folklore proof that symmetric properties can be tested with a constant number of samples.
Let $\calP$ be any symmetric property.
A tester can estimate the density $\E_{x \in \calX}[ f(x) ]$ up to additive accuracy
$\gamma$ for any small $\gamma > 0$ with a constant number of samples. 
This estimated density can be used to accept or reject the function based on
how close it is to the density of the functions in $\calP$.
The validity of this tester is established by showing that a function
can be $\epsilon$-far from $\calP$ only when its density is far from 
the density of every function in $\calP$.

Consider now a property $\calP$ that is $k$-part symmetric for some constant $k$.
Let $X_1,\ldots,X_k$ be a partition of $\calX$ associated with $\calP$.
We show that a tester which estimates the densities 
$\mu_i(f) := \E_{x \in X_i}[ f(x) ]$ for each $i=1,\ldots,k$
and then uses these densities to accept or reject is a valid tester for $\calP$.
We do so by showing that any function that is $\epsilon$-far from $\calP$ must have
a density vector that is far from those of every function in $\calP$.

\paragraph{Testability implies symmetry.}
To establish the second part of the theorem, we want to show that the existence
of a constant-sample tester $\calT$ for a property $\calP$ implies that
there is a partition of $\calX$ into a constant number of parts for which
$\calP$ is nearly determined by the density of functions within those parts.
We do so by using a variant of the 
Frieze--Kannan weak regularity lemma~\cite{ FriezeK96} for hypergraphs.
An \emph{$s$-uniform weighted hypergraph} is a hypergraph $G = (V,E)$ on $|V|$ vertices where $E : V^s \to [0,1]$ denotes the weight associated with each
hyperedge. Given a partition $V_1,\ldots,V_k$ of $V$ and a multi-index $I = (i_1,\ldots,i_s)$ with $i_1,\ldots,i_s \in [k]$, the \emph{expected weight of hyperedges of $G$ in $V_I$} is
$w_G(V_I) = w_G(V_{i_1},\ldots,V_{i_s}) = |\{ E(v_1,\ldots,v_s) : v_1 \in V_{i_1},\ldots, v_s \in V_{i_s}\}|/\prod_{i \in I} |V_i|$. Furthermore, for any subset $S \subseteq V$, we define $S \cap V_I = (S \cap V_{i_1}, \ldots, S \cap V_{i_s})$.

\begin{lemma}[Weak regularity lemma]
\label{lem:regularity}
  For every $\epsilon > 0$ and every $s$-uniform weighted hypergraph $G = (V,E)$, there is a partition $V_1,\ldots,V_k$ of $V$ with 
  $k = 2^{O(\log(\frac1\epsilon)/\epsilon^2)}$ parts
  such that for every subset $S \subseteq V$,
  \begin{equation}
    \sum_{I \in [k]^s} \frac{\prod_{i \in I} |S \cap V_i|}{|V|^s}
      \big| w_G(S \cap V_I) - w_G(V_I) \big| \le \epsilon.
  \end{equation}
\end{lemma}

This specific formulation of the weak regularity lemma seems not to have appeared
previously in the literature, but its proof is essentially the same as that of
usual formulations of the weak regularity lemma.
For completeness, we provide a proof of Lemma~\ref{lem:regularity} in 
Section~\ref{sec:regularity}.

Lemma~\ref{lem:regularity} is best described informally when we consider the special case where we consider unweighted graphs. 
In this setting, the weak regularity lemma says that 
for every graph $G$, there is a partition of the vertices of $G$ into $k=O(1)$ parts
$V_1,\ldots,V_k$ such that for 
every subset $S$ of vertices, the density of edges between 
$S \cap V_i$ and $S \cap V_j$ is close to the density between $V_i$ and $V_j$
\emph{on average} over the choice of $V_i$ and $V_j$. 
Regularity lemmas where this
density-closeness condition is satisfied for \emph{almost all} pairs of parts
$V_i$ and $V_j$ are known as ``strong'' regularity lemmas.
To the best of our knowledge, all previous characterization results in property
testing that relied on regularity lemmas (e.g.,~\cite{Alon:2009gn,Yoshida:2014tq,Bhattacharyya:2013fa}) used strong regularity lemmas.  This approach unavoidably
introduces tower-type dependencies between the query complexity and the 
characterization parameters.
By using a weak regularity lemma instead, we get a much better (though still 
doubly-exponential) dependence between the sample complexity and the 
partial symmetry parameter.

The second point of departure of our proof from previous characterizations is
in how we use the regularity lemma. In the prior work, the regularity lemma
was applied to the tested object itself (e.g., the dense graphs being tested
in~\cite{Alon:2009gn}) and the testability of the property was used to show that
the objects with the given property could be described by some combinatorial
characteristics related to the regular partition whose existence is promised by the regularity lemma.
Instead, in our proof of Theorem~\ref{thm:characterization}, we apply 
Lemma~\ref{lem:regularity} to a hypergraph associated with the \emph{tester}
itself, not with the tested object.

Specifically, let $\calT$ be an $s$-sample $\epsilon$-tester for some
property $\calP \subseteq \bit^\calX$. We associate $\calT$ with an
$s$-uniform weighted hypergraph $G_\calT$ on the set of vertices
$\calX \times \bit$.  The weight of each $s$-hyperedge of $G_\calT$ is
the acceptance probability of $\calT$ when its $s$ observations correspond
to the $s$ vertices covered by the hyperedge. 
By associating each function $f : \calX \to \bit$ with the subset
$S \subseteq \calX \times \bit$ that includes all $2^{\calX}$ vertices 
of the form $(x,f(x))$, we see that
the probability that $\calT$ accepts $f$ is the expected value of 
a hyperedge whose $s$ vertices are drawn uniformly and independently at random
from the set $S$. 
We can use Lemma~\ref{lem:regularity} to show that there is a 
partition of $V$ into a constant number of parts such that for each 
function $f$ with associated set $S$, this probability
is well approximated by some function of the density of $S$ in each of the 
parts.
We then use this promised partition of $V$ to partition the original 
input domain $\calX$ into a constant number of parts where membership in 
$\calP$ is essentially determined by the density of a function in each of 
these parts, as required.

\subsection{Proof that symmetry implies testability}
\label{sec:sufficiency}

We now begin the proof of Theorem~\ref{thm:characterization} with the easy
direction.

\begin{lemma}
\label{lem:if-part}
  Let $\calP \subseteq \bit^\calX$ be a property where for every $\epsilon > 0$, 
  there exists a constant $k = k_\calP(\epsilon)$ that is independent of $\calX$ and 
  a $k$-part symmetric property $\calP'$ such that 
  $\calP \subseteq \calP' \subseteq \calP_{\epsilon}$.
  Then $\calP$ is constant-sample testable.
\end{lemma}

\begin{proof}
Fix $\epsilon > 0$.
We show that we can distinguish functions in $\calP$ from functions that are $\epsilon$-far from $\calP$ with a constant number of samples.
From the condition in the statement of the lemma, there exists a $k$-part symmetric property $\calP'$ with $\calP \subseteq \calP' \subseteq \calP_{\epsilon/2}$ for some $k = k_{\calP}(\epsilon/2)$.
Let $S_1,\ldots,S_k$ be a partition of $\calX$ such that whether a function $f:\calX \to \bit$ satisfies $\calP'$ is determined by $\mu_{S_1}(f),\ldots,\mu_{S_k}(f)$.
For a set $S \subseteq \calX$, let $c_{S}(f) = \mu_{S}(f)|S|$ be the number of $x \in S$ with $f(x) = 1$.

Our algorithm for testing $\calP$ is as follows.
For each $i \in [k]$, we draw $q := O(k^2 \log k/\epsilon^2)$ samples $x_1,\ldots,x_q$ and computes the estimates $\widetilde{c}_{S_i}(f) := \frac{|\calX|}{q} \sum_{j \in [q]: x_j \in S_i} f(x_j)$ for each $i \in [k]$.
We accept if there exists $g \in \calP'$ such that
\[
  \sum_{i\in[k]} |\widetilde{c}_{S_i}(f) - c_{S_i}(g)| < \frac{\epsilon}{4} |\calX|,
\]
and reject otherwise.

Let us now establish the correctness of the algorithm.
By Hoeffding's bound, for each $i \in [k]$, we have $|c_{S_i}(f) - \widetilde{c}_{S_i}(f)| < \frac{\epsilon}{4k} |\calX|$ with probability at least $1-\frac1{3k}$ by choosing the hidden constant in the definition of $q$ sufficiently large.
By union bound, with probability at least $2/3$, we have $|c_{S_i}(f) - \widetilde{c}_{S_i}(f)| < \frac{\epsilon}{4k} |\calX|$ for every $i \in [k]$.
In what follows, we assume this happens.

If $f \in \calP$,
then the algorithm accepts $f$ because $\sum_{i\in[k]} |\widetilde{c}_{S_i}(f) - c_{S_i}(f)| < \frac\epsilon 4 |\calX|$ and $f \in \calP \subseteq \calP'$.

If $f$ is $\epsilon$-far from satisfying $\calP$,
then for any $g \in \calP'$, the triangle inequality and the fact that $\calP' \subseteq \calP_{\epsilon/2}$ imply that
\[
  \sum_{i\in[k]} |\widetilde{c}_{S_i}(f) - c_{S_i}(g)|
  \geq
  \sum_{i\in[k]} \Bigl(|c_{S_i}(f) - c_{S_i}(g)| - |\widetilde{c}_{S_i}(f) - c_{S_i}(f)|\Bigr)
  >
  \frac{\epsilon}{2} |\calX| - \frac{\epsilon}{4} |\calX|
  = \frac{\epsilon}{4}|\calX|
\]
and, therefore, the algorithm rejects $f$.
\end{proof}

\subsection{Testability implies symmetry}
\label{sec:necessity}

Suppose that a property $\calP$ is testable by a tester $\calT$ with sample complexity $s$.
We want to show that for any $\epsilon > 0$, there exists a $k$-part symmetric property $\calP'$ for $k = k(\epsilon)$ such that $\calP \subseteq \calP' \subseteq \calP_{\epsilon}$

For any $\bx = (x_1,\ldots,x_s) \in \calX^s$, we define $f(\bx) = \big( f(x_1),\ldots, f(x_s)\big)$
and we let $T(\bx,f(\bx)) \in [0,1]$ denote the acceptance probability of
the tester $\calT$ of $f$ when the samples drawn are $\bx$.
The overall acceptance probability of $f$ by $\calT$ is
\begin{align}
  \label{eqn:prob_accept}
  p_{\calT}(f) = \E_{\bx}[ T(\bx, f(\bx)) ].
\end{align}

We show that there is a family $\calS$ of a constant number of subsets of $\calX$ such that, for every
function $f : \calX \to \{0,1\}$, the acceptance probability $p_\calT(f)$ is
almost completely determined by the density of $f$ on the subsets in $\calS$.

\begin{lemma}
\label{lem:main}
  For any $\epsilon > 0$ and any $s$-sample tester $\calT$, there is a family
  $\calS = \{S_1,\ldots,S_m\}$ of $m \le 2^{O(2^{2s}/\epsilon^2)}$ subsets of
  $\calX$ such that for every $f : \calX \to \{0,1\}$,
  \begin{align*}
    \bigl| p_{\calT}(f) - \varphi_\calT(\mu_{S_1}(f),\ldots,\mu_{S_m}(f)) \bigr|
    \le \epsilon
  \end{align*}
  where $\mu_S(f) = \E_{x \in S}[ f(x) ]$ is the density of $f$ on the subset
  $S \subseteq \calX$ and $\varphi_\calT : [0,1]^m \to [0,1]$ is some fixed function.
\end{lemma}

\begin{proof}
  We consider the weighted hypergraph $G = (V, E)$ where $V = \calX \times \bit$ and $E$ is constructed by adding a hyperedge $((x_1,y_1),\ldots,(x_s,y_s))$ of weight $T(\bx,\by)$ for each $\bx = (x_1,\ldots,x_s) \in \calX^s$ and $\by = (y_1,\ldots,y_s)\in \bit^s$.

  A function $f : \calX \to \{0,1\}$ corresponds to a subset $S \subseteq V$ of
  size $|\calX| = |V|/2$, that is, $S := \set{(x,f(x)) \mid x \in \calX}$.
  The probability that $\calT$ accepts $f$ is
  \begin{align*}
    p_\calT(f) &= \E_{\bx \in \calX^s}[ T(\bx,f(\bx))]
    = \E_{\bv \in V^s}[ E(\bv) \mid v \in S^s]
    = \sum_{I \in [k]^s} \frac{\prod_{i \in I} |V_i \cap S|}{|S|^s} w_G(S \cap V_I) \\
    &= \sum_{I \in [k]^s} 2^s \frac{\prod_{i \in I} |V_i \cap S|}{|V|^s} w_G(S \cap V_I)
    = \sum_{I \in [k]^s} 2^s \frac{\prod_{i \in I} |V_i \cap S|}{|V|^s} w_G(V_I) \pm \epsilon
  \end{align*}
  where the last step is by Lemma~\ref{lem:regularity} applied with approximation
  parameter $\epsilon/2^s$.
  For every part $V_i$, let $V_i^1 = \{ x \in \calX : (x,1) \in V_i\}$ and let
  $V_i^0 = \{ x \in \calX : (x,0) \in V_i \}$. Then the
  value of $\sum_{I \in [k]^s} 2^s \frac{\prod_{i \in I} |V_i \cap S|}{|V|^s} w(V_I)$
  is completely determined by the density of $f$ on $V_1^1,V_1^0,V_2^1,V_2^0,\ldots,
  V_k^1,V_k^0$.
\end{proof}

We are now ready to complete the second part of the proof of Theorem~\ref{thm:characterization}.

\begin{lemma}
\label{lem:only-if-part}
  Let $\calP \subseteq \bit^\calX$ be constant-sample testable.
  Then for every $\epsilon > 0$, 
  there exists a constant $k = k_\calP(\epsilon)$ that is independent of $\calX$ and 
  a $k$-part symmetric property $\calP'$ such that 
  $\calP \subseteq \calP' \subseteq \calP_{\epsilon}$.
\end{lemma}

\begin{proof}
Fix any $\epsilon > 0$ and 
let $\calT$ be an $s$-sample $\epsilon$-tester for $\calP$.
Let $\gamma < \frac13$ be any constant that is less than $\frac13$.
By Lemma~\ref{lem:main} applied with the parameter $\gamma$,
there is a family $\calS = \set{S_1,\ldots,S_m}$ with $m = 2^{O(2^{2s}/\gamma^2)}$ 
sets such that for every $f:\calX \to \bit$,
  \begin{align}
    |p_{\calT}(f) - \varphi_{\calT}(\mu_{S_1} (f), \ldots , \mu_{S_k} (f))| \leq \gamma. \label{eq:well-approximation}
  \end{align}
  Define
  \[
    \calP' = \set{f: \calX \to \bit: \exists g \in \calP \text{ s.t. } (\mu_{S_1}(f),\ldots,\mu_{S_k}(f)) = (\mu_{S_1}(g),\ldots,\mu_{S_k}(g))}.
  \]
  This construction trivially guarantees that $\calP' \supseteq \calP$.
  Furthermore,~\eqref{eq:well-approximation} guarantees that for every $f \in \calP'$, if we let $g\in \calP$ be one of the elements with the same density profile as $f$,
  \[
    p_{\calT}(f) \geq \varphi_{\calT}(\mu_{S_1} (f), \ldots , \mu_{S_k} (f)) - \gamma = \varphi_{\calT}(\mu_{S_1} (g), \ldots , \mu_{S_k} (g)) - \gamma \geq p_\calT(g) - 2\gamma.
  \]
  Since $\gamma < 1/3$, the fact that $\calT$ is an $\epsilon$-tester for $\calP$ and that $g \in \calP$ imply that $f \in \calP_\epsilon$.

  Let $S'_1,\ldots,S'_k$ be the family of sets obtained by taking intersections and complements of $S_1,\ldots,S_m$.
  Note that $S'_1,\ldots,S'_k$ forms a partition of $\calX$ and $\mu_{S_1},\ldots,\mu_{S_m}$ is completely determined by $\mu_{S'_1},\ldots,\mu_{S'_k}$.
  Furthermore, $k = O(2^m)$.
  Hence, $\calP'$ is a $k$-part symmetric property induced by the partition $S'_1,\ldots,S'_k$ with $\calP \subseteq \calP' \subseteq \calP_\epsilon$.
\end{proof}

Theorem~\ref{thm:characterization} follows immediately from Lemmas~\ref{lem:if-part}
and~\ref{lem:only-if-part}.
%!TEX root=SampleTestableCharacterization.tex

\section{Applications}
\label{sec:applications}

\subsection{Affine-invariant properties}

For an affine transformation $A: \F_p^n \to \F_p^n$ and a function $f:\F_p^n \to \bit$, we define $Af:\F_p^n \to \bit$ to be the function that satisfies $Af(x) = f(Ax)$
for every $x \in \F_p^n$.
A property $\calP$ of functions $f:\F_p^n \to \bit$ is \emph{affine-invariant} if for any affine transformation $A:\F_p^n \to \F_p^n$ and $f \in \calP$, we have $Af \in \calP$.
Our characterization shows that the only affine-invariant properties of functions
$f : \F_p^n \to \{0,1\}$ that are testable with a constant number of samples are
the (fully symmetric) properties that are determined by the density of $f$.

\newtheorem*{affinecor}{Corollary~\ref{cor:affine}}
\begin{affinecor}[Restated]
  For every $\epsilon > 0$, if $\calP$ is an affine-invariant property of functions $f : \F_p^n \to \{0,1\}$ that
  is $\epsilon$-testable with a constant number of samples, then there is a symmetric property
  $\calP'$ such that $\calP \subseteq \calP' \subseteq \calP_\epsilon$.
\end{affinecor}

\begin{proof}
  By Theorem~\ref{thm:characterization}, if $\calP$ is testable with $O(1)$ samples, then there
  are subsets $S_1,\ldots,S_k$ of $\F_p^n$ with $k = O(1)$ and a property
  $\calP''$ such that $\calP \subseteq \calP'' \subseteq \calP_\epsilon$ and
  $\calP''$ is invariant under all permutations of $S_1,\ldots,S_k$.
  Let $\calP'$ be the closure of $\calP''$ under all affine transformations
  over $\F_p^n$. (I.e., $\calP' = \{Af : f \in \calP'', A \mbox{ is affine}\}$.)
  Since $\calP$ itself is invariant under affine transformations, we
  have that $\calP' \subseteq \calP_\epsilon$.
  We want to show that $\calP'$ is symmetric.
  To show this, it suffices to show that $\calP'$ is closed under
  transpositions.
  I.e., that for every $x,y \in \F_p^n$ and $f \in \calP'$, the function
  $g$ obtained by setting $g(x) = f(y)$, $g(y) = f(x)$, and $g(z) = f(z)$ for
  every other $z \notin \{x,y\}$ is also in $\calP'$.
  We write $g = (x\,y)f$ to denote the action of the transposition $(x\,y)$ on
  $f$.

  If $x,y \in S_i$ for some $i \in [k]$, then our conclusion follows immediately
  from the invariance of $\calP''$ over permutations on $S_i$. Otherwise, let
  $w,z \in S_j$ for some $j \in [k]$. Since the set of affine transformations
  is a doubly-transitive action on $\F_p^n$,
  there is a transformation $A$ such that $A(x) = w$ and $A(y) = z$. Then
  $(x\,y)f = A^{-1}(w\,z)Af \in \calP'$, as we wanted to show.
\end{proof}

\subsection{Graph properties}

Similarly, our characterization shows that the only graph properties that are
testable with a constant number of samples are
the (fully symmetric) properties that are determined by the edge density of the
graph.

\newtheorem*{graphcor}{Corollary~\ref{cor:graph}}
\begin{graphcor}[Restated]
  For every $\epsilon > 0$, if $\calP$ is a graph property that
  is $\epsilon$-testable with a constant number of samples, then there is a symmetric property
  $\calP'$ such that $\calP \subseteq \calP' \subseteq \calP_\epsilon$.
\end{graphcor}

\begin{proof}
  The proof is nearly identical.
  By Theorem~\ref{thm:characterization}, if $\calP$ is testable with $O(1)$ samples, then there
  are subsets $S_1,\ldots,S_k$ of the edge set with $k = O(1)$ and a property
  $\calP''$ such that $\calP \subseteq \calP'' \subseteq \calP_\epsilon$ and
  $\calP''$ is invariant under all permutations of $S_1,\ldots,S_k$.
  Let $\calP'$ be the closure of $\calP''$ under all permutations of the
  vertex set.
  We again have that $\calP \subseteq \calP' \subseteq \calP_\epsilon$ and
  we want to show that $\calP'$ is invariant under every transposition of the
  edge set.

  The one change with the affine-invariant property proof is that the set of
  permutations of the vertices of a graph is not a 2-transitive action on the
  set of edges of this graph. But the same idea still works because we can
  always find a vertex permutation to send two edges on disjoint vertices
  to a same part $S_i$ that also contains a pair of edges on disjoint vertices,
  and we can find a vertex permutation to send two edges that share
  a common vertex to a part $S_j$ that also contains two edges that
  share a common vertex. So for every pair of edges $e_1, e_2$, we have
  a vertex permutation $\pi_V$ and a pair of edges $e_3,e_4 \in S_\ell$
  such that if $G \in \calP'$, then
  $(e_1\,e_2)G = \pi_V^{\,-1} (e_3\,e_4) \pi_V G \in \calP'$.
\end{proof}

\subsection{Testing monotonicity}

With Theorem~\ref{thm:characterization},
to show that monotonicity of functions
$f : [n]^d \to \bit$ is constant-sample testable,
it suffices to identify an $O(1)$-part symmetric function that
covers all the monotone functions and does not include any
function that is far from monotone. This is what we do below.

\newtheorem*{monocor}{Corollary~\ref{cor:monotonicity}}
\begin{monocor}[Restated]
  For every constant $d \ge 1$ and constant $\epsilon > 0$, we can $\epsilon$-test monotonicity of
  functions $f : [n]^d \to \{0,1\}$ on the $d$-dimensional hypergrid
  with a constant number of samples.
\end{monocor}

\begin{proof}
For any $\epsilon > 0$,
fix $k = \lceil d/\epsilon \rceil$ and
let $\calR$ be a partition of the space $[n]^d$ into $k^d$ subgrids of side length
(at most) $\lfloor \epsilon n\rfloor$ each.
We identify the parts in $\calR$ with the points in $[k]^d$.
For an input $x \in [n]^d$, let
$\phi_\calR(x)$ denote the part of $\calR$ that contains $x$.

Given some function $f : [n]^d \to \{0,1\}$, define the \emph{$\calR$-granular}
representation of $f$ to be the function $f_\calR : [k]^d \to \{0,1,*\}$ defined by
$$
f_\calR(x) =
\begin{cases}
0 & \mbox{if } \forall y \in [n]^d \text{ with } \phi_\calR(y) = x, f(y) = 0 \\
1 & \mbox{if } \forall y \in [n]^d \text{ with } \phi_\calR(y) = x, f(y) = 1 \\
* & \mbox{otherwise.}
\end{cases}
$$
Let $\calP = \{ f : [n]^d \to \{0,1\} : \exists \mbox{ monotone } g \mbox{ s.t. }
f_\calR = g_\calR \}$ be the property that includes every function whose
$\calR$-granular representation equals that of a monotone function.
By construction $\calP$ includes all the monotone functions and is invariant under
any permutations within the $O(1)$ parts of $\calR$. To complete the proof
of the theorem, we want to show that every function in $\calP$ is
$\epsilon$-close to monotone.

Fix any $f \in \calP$ and let $g : [n]^d \to \bit$ be a monotone function for which
$f_\calR = g_\calR$. The distance between $f$ and $g$ is bounded by
$$
\dist(f,g) \le \frac{|g^{-1}_{\calR}(*)|}{k^d}.
$$
Now consider the poset $P$ on $[k]^d$ where $x \prec y$ iff $x_i < y_i$ for every
$i \in [d]$.
We first observe that the set $g^{-1}_{\calR}(*)$ forms an anti-chain on this poset.
Indeed, if there exist $x, y \in [k]^d$ with $x \prec y$ and $g(x) = g(y) = *$, then there exist $x',y' \in [n]^d$ such that $\phi_R(x') = x$, $\phi_{R}(y')= y$, $g(x') =1$, and $g(y')=0$.
But this contradicts the monotonicity of $g$ because $x' \leq y'$ holds from $x \prec y$.

For $x \in [k]^d$, let $x_{\max} = \max_{i \in [d]} x_i$ and $x_{\min} = \min_{i \in [d]} x_i$. Define $\mathbf{1} = (1,1,1,\ldots,1) \in [k]^d$
and $S = \{ x \in [k]^d : x_{\min} = 1\}$. We can partition $[k]^d$
into $|S| = k^d - (k-1)^d \le d k^{d-1}$ chains $(x, x+\mathbf{1}, x + 2 \cdot \mathbf{1}, \ldots, x + (x_{\max}-1) \cdot \mathbf{1})$, one for each $x \in S$.
Therefore, by Dilworth's theorem,
every anti-chain on $P$ has size at most $d k^{d-1}$.
In particular, this bound holds for the anti-chain $g^{-1}(*)$ so $\dist(f,g) \le \frac{dk^{d-1}}{k^d} = \frac dk \le \epsilon$.
\end{proof}

%!TEX root=SampleTestableCharacterization.tex

\section{Proof of the weak regularity lemma}\label{sec:regularity}

\subsection{Information theory}

The proof we provide for Lemma~\ref{lem:regularity} is information-theoretic.
In this section, we will use bold fonts to denote random variables.
We assume that the reader is familiar with the basic concepts of 
entropy and mutual information; a good
introduction to these definitions is~\cite{CoverT12}.
The only facts we use about these concepts is the chain rule for 
mutual information and the fact that the entropy of a random
variable is at most the logarithm of the number of values it can take.

The one non-basic information-theoretic inequality that we use in the proof is 
an inequality established by Tao~\cite{ Tao06} and later refined by Ahlswede~\cite{ Ahlswede07}.

\begin{lemma}[Tao~\cite{Tao06}, Ahlswede~\cite{Ahlswede07}]
\label{lem:Tao}
  Let $\by$, $\bz$, and $\bz'$ be discrete random variables where $\by \in [-1,1]$ and $\bz' = \phi(\bz)$ for some function $\phi$. Then
  \[
    \E\Big[ \big| \E[\by \mid \bz' ] - \E[ \by \mid \bz ] \big| \Big] \le
    \sqrt{ 2 \ln 2 \cdot I( \by ; \bz \mid \bz')}.
  \]
\end{lemma}

Tao originally used his inequality to offer an information-theoretic proof
of Szemer{\'e}di's (strong) Regularity Lemma. The proof we offer below follows (a simplified version of) the same 
approach. The fact that Tao's proof of the strong regularity lemma can
also be applied (with simplifications) to prove the Frieze--Kannan weak 
regularity lemma was observed previously by Trevisan~\cite{Trevisan06}.

\subsection{Proof of Lemma~\ref{lem:regularity}}

For $\tau > 0$, we say that a hypergraph is \emph{$\tau$-granular} if the weight of each hyperedge is a multiple of $\tau$.
When proving Lemma~\ref{lem:regularity}, we can assume that the given graph is $\tau$-granular for $\tau = \Theta(\epsilon)$.
To see this, let $G' = (V,E')$ be the hypergraph obtained from $G = (V,E)$ by rounding the weight of each hyperedge to a multiple of $\tau$.
Then, for any set $S \subseteq V$ and a partition $V_1,\ldots,V_k$ of $V$, we have
\begin{align*}
  \sum_{I \in [k]^s} \frac{\prod_{i \in I} |S \cap V_i|}{|V|^s}
    \left| w_G(S \cap V_I) - w_G(V_I) \right|
  & =
  \sum_{I \in [k]^s} \frac{\prod_{i \in I} |S \cap V_i|}{|V|^s}
    \left| (w_{G'}(S \cap V_I) \pm \tau) - (w_{G'}(V_I) \pm \tau\right)| \\
  & =
  \sum_{I \in [k]^s} \frac{\prod_{i \in I} |S \cap V_i|}{|V|^s}
    \left| w_{G'}(S \cap V_I) - w_{G'}(V_I) \right| \pm 2\tau.
\end{align*}
In order to make the right-hand side less than $\epsilon$, it suffices to show that the claim of Lemma~\ref{lem:regularity} holds for $\epsilon/3$-granular hypergraphs with an error parameter $\epsilon/3$.
In what follows, we assume the input graph $G$ is $\epsilon/3$-granular.

Let $V_1,\ldots,V_\ell$ be any partition of $V$.
Draw $\bv \in V^s$ uniformly at random.
Define $\by = E(\bv)$.
Let $\psi : V^s \to [\ell]^s$ be the function that identifies the
parts containing each of $s$ vertices.
For any set $S \subseteq V$, let $1_S : V^s \to \{0,1\}^s$ be the
indicator function of $S$ for $s$-tuples of vertices.
Define $\bz_S = (\psi(\bv), 1_S(\bv))$ and
$\bz' = \psi(\bv)$.
We consider two cases.

First, consider the situation where for every set $S \subseteq V$,
\begin{equation}
\label{eq:case1}
I( \by; \bz_S \mid \bz') \le \frac{(\epsilon/3)^2}{2\ln 2}.
\end{equation}
Then by Tao's lemma, for every set $S$ we also have
\begin{align*}
\sum_{I \in [k]^s, b \in \{0,1\}^s} \frac{\prod_{j \in [s]} |S_{b_j} \cap V_{i_j}|}{|V|^s} \left| w_G( S_b \cap V_I ) - w_G(V_I) \right|
&= |\E[ E(\bv) \mid \psi(\bv) ] - \E[ E(\bv) \mid \psi(\bv), 1_S(\bv)]| \\
&= |\E[ \by \mid \bz' ] - \E[ \by \mid \bz_S ]|
\le \epsilon/3.
\end{align*}
And clearly
$$
\sum_{I \in [k]^s} \frac{\prod_{j \in [s]} |S \cap V_{i_j}|}{|V|^s} \left| w_G( S \cap V_I ) - w_G(V_I) \right|
\le
\sum_{I \in [k]^s, b \in \{0,1\}^s} \frac{\prod_{j \in [s]} |S_{b_j} \cap V_{i_j}|}{|V|^s} \left| w_G( S_b \cap V_I ) - w_G(V_I) \right|
$$
since the expression in the left-hand side is one of the terms in the sum (over $b$) on the right-hand side.
So in this case the lemma holds.

Second, we need to consider the case where there is some set $S \subseteq V$
for which
\begin{equation}
\label{eq:case2}
I( \by; \bz_S \mid \bz') > \frac{(\epsilon/3)^2}{2\ln 2}.
\end{equation}
Then by the chain rule for mutual information
$$
I( \by ; \bz_S) = I(\by ; \bz') + I( \by; \bz_S \mid \bz')
\ge I(\by ; \bz') + \frac{(\epsilon/3)^2}{2\ln 2}.
$$
Define the \emph{information value} of a partition with indicator function
$\psi$ as $I(E(\bv) ; \psi(\bv))$. Then the above observation shows that
when~\eqref{eq:case2} holds, we can obtain a refined partition with $2\ell$
parts whose information value increases by at least $\frac{(\epsilon/3)^2}{2 \ln 2}$.
The information value of any partition is bounded above by $H(\by)$, so after at most $\frac{2\ln 2 \cdot  H(\by)}{(\epsilon/3)^2}$ refinements, we must obtain a partition
(with at most $2^{\frac{2 \ln 2 \cdot H(\by)}{(\epsilon/3)^2}}$ parts) that
satisfies~\eqref{eq:case1}.
Since $G$ is $\epsilon/3$-granular, we have $H(\by) = O(\log (1/\epsilon))$, and the lemma follows.
\qed

\bibliographystyle{plain}
\bibliography{SampleTestable}

\end{document}